\documentclass[sigconf,nonacm]{acmart}
\usepackage{graphicx} 
\usepackage{tabularx}
\usepackage{ragged2e}

\setcitestyle{numbers,sort&compress}

\usepackage{amsmath}

\title[Proof of Commitment]{Proof of Commitment: A Human-Centric Resource for Permissionless Consensus}

\author{Homayoun Maleki}
\affiliation{%
	\institution{University of Deusto}
	\department{DeustoTech}
	\city{Bilbao}
	\country{Spain}
}
\email{h.maleki@deusto.es}
\author{Nekane Sainz}
\affiliation{%
	\institution{University of Deusto}
	\department{DeustoTech}
	\city{Bilbao}
	\country{Spain}
}
\email{nekane.sainz@deusto.es}
\author{Jon Legarda}
\affiliation{%
	\institution{University of Deusto}
	\department{DeustoTech}
	\city{Bilbao}
	\country{Spain}
}
\email{jlegarda@deusto.es}

\begin{document}

\begin{abstract}
	Permissionless consensus protocols require a scarce resource to regulate
	leader election and to provide Sybil resistance. Existing paradigms---Proof of
	Work (PoW) and Proof of Stake (PoS)---instantiate this scarcity through
	\emph{parallelizable} resources such as computational throughput or financial
	capital. Once acquired, these resources can be subdivided across arbitrarily
	many identities at negligible marginal cost, implying a fundamental
	impossibility of enforcing \emph{linear} Sybil cost in PoW/PoS-style systems.
	
	We introduce \emph{Proof of Commitment (PoCmt)}, a consensus primitive grounded
	in a fundamentally \emph{non-parallelizable} resource: real-time human
	engagement. Each validator maintains a decomposable commitment state capturing
	cumulative human effort, protocol participation, and online availability.
	Engagement is enforced via a Human Challenge Oracle (HCO) that issues
	time-sensitive, human-verifiable tasks whose solutions are identity-bound and
	must be produced within a bounded time window. As a result, the number of
	challenges that can be solved per epoch is intrinsically limited by available
	human-time, independent of identity count or capital.
	
	Under this model, sustaining $s$ active identities requires $\Theta(s)$ units of
	human-time per \emph{human window}. We establish a sharp cost-theoretic separation: any
	protocol whose weighting resource is parallelizable admits asymptotically zero
	marginal Sybil cost, whereas PoCmt enforces a strictly linear Sybil-cost profile.
	Building on a weighted-backbone analysis, we show that PoCmt achieves safety,
	liveness, and commitment-proportional fairness under partial synchrony, and that
	adversarial commitment cannot outgrow honest commitment without sustained
	linear human effort.
	
	We complement the analysis with simulations that isolate human-time capacity as
	the sole adversarial bottleneck. The results empirically validate the commitment
	drift invariant, demonstrate that increasing identity count or capital alone
	does not amplify influence, and show that availability decay suppresses dormant
	or rotating identities. Together, these results position PoCmt as a new point in
	the consensus design space, grounding permissionless security in human-time
	rather than computation or capital.
\end{abstract}

	\maketitle
	

\section{Introduction}

Permissionless consensus in open networks requires a scarce and verifiable
resource to regulate leader election and to provide Sybil resistance.
Classical paradigms instantiate this resource as \emph{computational power}
(Proof of Work, PoW)~\cite{nakamoto2008bitcoin} or \emph{monetary capital}
(Proof of Stake, PoS)~\cite{garay2015bitcoin}. While successful in practice,
both paradigms rely on resources that are inherently \emph{parallelizable}:
once acquired, hashpower or stake can be subdivided across arbitrarily many
identities at negligible marginal cost. Consequently, PoW/PoS systems cannot
enforce \emph{linear} Sybil cost: sustaining many identities is asymptotically
no more expensive than sustaining one, and identity replication becomes a
structural avenue for adversarial amplification.

\smallskip
\noindent\textbf{Beyond parallelizable resources.}
This paper explores a fundamentally different class of scarcity—one that is
\emph{intrinsically non-parallelizable}, resistant to automation, and tied to
individual human trajectories. Our key observation is that \emph{human time}
and \emph{cognitive effort} possess precisely these properties: each unit of
human-verifiable work must be performed by a single human, cannot be
fractionally delegated, and cannot be parallelized beyond the number of humans
actively engaged within a given operational window.

\smallskip
\noindent\textbf{Proof of Commitment (PoCmt).}
We introduce \emph{Proof of Commitment (PoCmt)}, a permissionless consensus
primitive in which validator influence is derived from \emph{real-time,
	human-verifiable engagement} rather than machines or capital. Each validator
maintains a three-dimensional commitment state
\[
S_v(t)=\big(H_v(t),\,P_v(t),\,U_v(t)\big),
\]
capturing cumulative human engagement, protocol participation, and online
availability. These components evolve through deterministic boost, decay, and
slashing rules, yielding a time-dependent commitment score that directly
determines leader-election weight.

Human effort is formalized via a \emph{Human Challenge Oracle} (HCO), which
issues identity-bound, time-sensitive challenges that must be solved within a
bounded epoch. Crucially, at most a bounded number of challenges can be solved
per epoch, independent of the number of identities controlled by an adversary.
Solving one validator’s challenge provides no advantage for another, and
automated solvers succeed only with negligible probability. As a result,
maintaining $s$ actively engaged identities necessarily requires
$\Theta(s)$ units of human-time per human windows.

\smallskip
\noindent\textbf{A new Sybil-resistance regime.}
PoCmt induces a cost-theoretic separation between consensus protocols based on
parallelizable resources and those grounded in human-time. In particular, we
show that:
\begin{itemize}
	\item any protocol whose weighting resource is parallelizable (e.g.,
	hashpower or stake) admits \emph{zero marginal Sybil cost}, as identities can
	be replicated without additional resource expenditure;
	
	\item PoCmt enforces \emph{linear Sybil cost}: sustaining $s$ identities for
	$T$ epochs requires $\Theta(sT)$ units of human effort, and adversarial
	commitment cannot outpace honest commitment without continuously paying
	this cost;
	
	\item commitment-weighted leader election in PoCmt satisfies backbone-style
	safety, liveness, and proportional fairness under standard partial synchrony
	assumptions.
\end{itemize}

\smallskip
\noindent\textbf{Why human-time matters.}
Human involvement has traditionally appeared only as a one-time identity gate,
such as CAPTCHAs or Proof-of-Personhood ceremonies. PoCmt elevates human-time to
a \emph{persistent, quantifiable, consensus-relevant resource}. This shift
fundamentally alters the Sybil-resistance landscape: adversarial influence is no
longer limited by hardware, capital, or identity creation, but by sustained
human effort—a resource that is scarce, economic, and inherently
non-parallelizable.

\paragraph*{Contributions.}
\begin{itemize}
	\item We introduce PoCmt, the first consensus primitive grounded in
	non-parallelizable human-time, formalized through a decomposable commitment
	state.
	
	\item We define a Human Challenge Oracle (HCO) capturing identity-bound,
	AI-resistant engagement with explicit real-time constraints.
	
	\item We prove a cost-theoretic separation showing that
	parallelizable-resource protocols admit zero marginal Sybil cost, whereas
	PoCmt enforces $\Theta(sT)$ adversarial cost to sustain $s$ identities for
	$T$ epochs.
	
	\item Using a weighted-backbone analysis, we establish safety, liveness, and
	fairness under partial synchrony and characterize the human effort required
	for adversarial majority.
	
	\item We present simulations that isolate human-time capacity as the sole
	adversarial bottleneck, empirically validating commitment drift, fairness,
	and the suppression of dormant or rotating identities.
\end{itemize}

\smallskip
\noindent\textbf{Scope and incentives.}
Our analysis focuses on the security consequences of using human-time as the
weighting resource under the standard honest-but-participating assumption.
Modeling incentives for sustained participation—including fatigue, churn,
abandonment, and validator replacement—is orthogonal to the cost-theoretic
separation and is discussed explicitly as a deployment limitation in
Section~\ref{sec:discussion}.

PoCmt demonstrates that permissionless consensus need not be anchored solely to
machine-dominated resources. By grounding security in human-time, it identifies
a qualitatively new point in the consensus design space and expands the
theoretical foundations of Sybil-resistant blockchain protocols.

\section{Background and Related Work}
\label{sec:related}

\paragraph{Sybil resistance and resource-based consensus.}
The Sybil attack—creating many identities at negligible marginal cost—is a
foundational obstacle for open membership systems~\cite{douceur2002sybil}.
Permissionless consensus protocols address this challenge by tying influence to
an assumed-scarce resource. In Proof-of-Work (PoW), influence is proportional to
computational effort~\cite{nakamoto2008bitcoin}, while in Proof-of-Stake (PoS),
it is proportional to locked capital. A substantial body of work formalizes
safety and liveness of such protocols under backbone-style frameworks and related
models~\cite{garay2015bitcoin,bonneau2015bitcoin,pass2017snow}, and studies
strategic deviations including selfish mining~\cite{eyal2014selfish} and
high-rate or forking dynamics~\cite{sompolinsky2015secure}. Representative PoS
protocols with rigorous analyses include Ouroboros~\cite{kiayias2017ouroboros},
while related work considers participation churn and ``sleepy'' validators under
partial synchrony~\cite{pass2017snow}.

A fundamental limitation shared by PoW and PoS is that their weighting
resources—hashpower and stake—are economically \emph{parallelizable}: once
acquired, they can be subdivided across arbitrarily many identities at
negligible marginal cost. As a consequence, linear Sybil cost is unattainable in
such systems, a phenomenon formalized in our cost-theoretic analysis
(Section~\ref{sec:cost}).

\smallskip
To contextualize PoCmt within the broader landscape, Table~\ref{tab:resource-taxonomy}
summarizes prior Sybil-resistance approaches according to the nature of the
underlying resource, its degree of parallelizability, and the resulting Sybil
cost profile. Importantly, the table distinguishes between
\emph{consensus-weighting resources} and \emph{identity or access mechanisms},
highlighting that prior work has not combined non-parallelizable human effort
with a persistent, time-evolving consensus weight.

\begin{table*}[t]
	\centering
	\small
	\setlength{\tabcolsep}{6pt}
	\renewcommand{\arraystretch}{1.15}
	\begin{tabularx}{\textwidth}{l c c >{\RaggedRight\arraybackslash}X}
		\hline
		\textbf{Resource Type} &
		\textbf{Parallelizable} &
		\textbf{Sybil Cost Profile} &
		\textbf{Representative Systems} \\
		\hline
		Computational power & Yes & Sublinear / zero marginal &
		PoW (Bitcoin)~\cite{nakamoto2008bitcoin} \\
		
		Monetary capital & Yes & Sublinear / zero marginal &
		PoS; Ouroboros~\cite{kiayias2017ouroboros} \\
		
		Storage / space & Yes & Sublinear &
		SpaceMint~\cite{dodd2021spacemint} \\
		
		Trusted hardware / timers & Yes & Sublinear &
		PoET; SGX-based designs \\
		
		Social identity graphs & Partially & Bounded (non-economic) &
		SybilGuard~\cite{yu2008sybilguard} \\
		
		One-time human verification & No (one-shot) & Constant setup cost &
		CAPTCHAs~\cite{vonahn2003captcha}; Proof-of-Personhood~\cite{moser2019popp} \\
		\hline
		\textbf{Human-time (PoCmt)} & \textbf{No} & \textbf{Linear per epoch} &
		\textbf{PoCmt (this work)} \\
		\hline
	\end{tabularx}
	\caption{\textbf{Taxonomy of Sybil-resistance mechanisms.}
		Parallelizable resources permit cheap identity replication and therefore
		admit sublinear or zero marginal Sybil cost. Identity and access mechanisms
		(e.g., CAPTCHAs or Proof-of-Personhood) restrict identity creation but do not
		define persistent consensus weight. PoCmt uniquely derives leader-election
		power from a non-parallelizable, time-bounded human resource, enforcing linear
		Sybil cost per active identity per epoch.}
	\label{tab:resource-taxonomy}
\end{table*}

\paragraph{Leader election, committees, and BFT finality.}
PoCmt’s protocol layer builds on standard primitives for leader election and
finality. Commitment-weighted lotteries can be instantiated using verifiable
random functions (VRFs)~\cite{micali1999verifiable} and used to sample leaders and
committees, as in committee-based designs such as Algorand~\cite{micali2016algorand,gilad2017algorand}.
For deterministic finality, PoCmt is compatible with classical and modern BFT
protocols, including PBFT~\cite{castro1999practical} and HotStuff~\cite{yin2019hotstuff}.
Our contribution is orthogonal to these protocol mechanics: we introduce a new
\emph{weighting resource} and analyze the resulting safety and liveness
properties under partial synchrony~\cite{dwork1988consensus}.

\paragraph{Alternative resource classes.}
Beyond computation and stake, several proposals tie consensus influence to other
machine-controllable resources, including storage, bandwidth, or trusted
hardware. Proof-of-Space and storage-based schemes rely on disk capacity as the
scarce resource~\cite{dodd2021spacemint}, while trusted-hardware approaches
replace economic scarcity with assumptions about secure enclaves or timers.
Although these designs alter incentive profiles or energy costs, the underlying
resource remains \emph{machine-parallelizable} and scalable via capital
investment. In contrast, PoCmt targets a scarcity class whose limiting factor is
human-time per operational window, intrinsically bounded by the number of humans
actively engaged.

\paragraph{Identity-based and human-in-the-loop mechanisms.}
A distinct line of work seeks Sybil resistance by constraining or validating
identities rather than weighting them by a resource. Graph-based defenses exploit
social-network structure~\cite{yu2008sybilguard}, while Proof-of-Personhood
ceremonies aim to approximate ``one person, one identity'' under varying trust
assumptions~\cite{moser2019popp}. Human verification mechanisms such as
CAPTCHAs~\cite{vonahn2003captcha} exploit cognitive hardness gaps but are deployed
as one-time access gates. These approaches do not define a persistent,
time-evolving consensus weight and therefore do not integrate directly with
leader election or backbone-style safety and liveness analyses.

\paragraph{Relation to PoCmt.}
PoCmt is conceptually distinct from both machine-resource consensus (PoW/PoS and
their variants) and identity-centric Sybil defenses. It does not enforce
one-person–one-identity. Instead, it models \emph{human engagement} as a scarce,
non-parallelizable, time-bound resource that must be continually replenished.
Through repeated Human Challenge Oracle (HCO) tasks and commitment-state
evolution, PoCmt introduces a mathematically defined, human-centric weighting
signal that integrates directly with standard leader-election and security
analyses. To the best of our knowledge, prior work has not provided a formal
consensus framework in which \emph{human-time per epoch} is the fundamental
scarce resource governing leader-election power.

\section{Model and Human-Time Resource}
\label{sec:model}

We present the formal model underlying Proof of Commitment (PoCmt).
The central design goal is to ground consensus weight in a
\emph{non-parallelizable human-time resource}, while retaining
epoch-based leader election and standard backbone-style reasoning.

\subsection{Time Model: Epochs and Human Windows}
\label{sec:time-model}

PoCmt separates time into two coupled layers.

\paragraph{Consensus epochs.}
The protocol proceeds in discrete consensus epochs
$t = 0,1,2,\dots$, where each epoch corresponds to a fixed
wall-clock duration $\Delta_e$ (e.g., seconds or minutes).
Leader election, block proposal, and fork choice are executed
at this timescale.

\paragraph{Human windows.}
Human-verifiable challenges are issued at a coarser
\emph{human window} scale indexed by
$d = 0,1,2,\dots$ (e.g., days).
Each human window spans $\Delta_h$ time units and contains
$E = \Delta_h / \Delta_e$ consensus epochs.
We write $d(t)$ for the human window containing epoch $t$.

This two-timescale structure reflects the operational reality
that block production is frequent, while human engagement is
naturally rate-limited and typically performed only a small
number of times per day.

\subsection{Commitment State and Score}

Each validator $v$ maintains a commitment state
\[
S_v(t) = \big(H_v(t), P_v(t), U_v(t)\big),
\]
where:
\begin{itemize}
	\item $H_v(t)$ captures accumulated human engagement,
	\item $P_v(t)$ captures protocol participation and honesty,
	\item $U_v(t)$ captures online availability.
\end{itemize}

The commitment score used for leader election at epoch $t$ is
\[
CS_v(t) = \alpha H_v(t) + \beta P_v(t) + \gamma U_v(t),
\]
for fixed nonnegative weights $(\alpha,\beta,\gamma)$.

\subsection{Human Challenge Oracle (HCO)}
\label{sec:HCO}

PoCmt models human-time as a non-parallelizable resource via a
\emph{Human Challenge Oracle} (HCO).
Unlike per-epoch challenge mechanisms, HCO issues only a small
number of challenges per human window.

\paragraph{Challenge rate and difficulty.}
For each human window $d$, the protocol specifies a
\emph{challenge rate} $k(d) \in \mathbb{N}$.
This parameter serves as a window-level \emph{difficulty} knob:
each validator is expected to solve up to $k(d)$ challenges
during window $d$ to remain fully engaged.

For each validator $v$ and index $j \in \{1,\dots,k(d)\}$,
the oracle outputs a fresh challenge
\[
\chi_{v,d,j} \leftarrow \mathrm{HCO}(v,d,j).
\]

\paragraph{Oracle assumptions.}
HCO satisfies the following properties.

\begin{description}
	\item[(H1) Human advantage under time limits.]
	For the deployed challenge family at operational difficulty $d$,
	an honest human solves within $\Delta_{\mathrm{resp}}$ with high probability,
	while any feasible automated solver succeeds with probability at most
	$\varepsilon(d)$, where $\varepsilon(d)$ is sufficiently small for security.
	
	\item[(H2) Identity binding.]
	Each challenge is bound to $(v,d,j)$; solutions cannot be reused
	across validators or challenge indices.
	
	\item[(H3) Real-time constraint.]
	A solution must be submitted within a short response window
	$\Delta_{\mathrm{resp}}$, preventing precomputation and
	long-term stockpiling.
	
	\item[(H4) Human-time parallelism bound.]
	A single human can solve at most $\tau_h = O(1)$ challenges per human window,
	independent of the number of identities controlled.
	If an adversary employs $m$ humans, the total number of challenges that can
	be solved in a window is at most $m \cdot \tau_h$.
\end{description}
In the remainder of the paper, we define the adversarial human-time capacity
as $M = m \cdot \tau_h$, representing the total number of challenges that can
be solved per window. All bounds are stated with respect to $M$.

\paragraph{Parallelism bound.}
If an adversary controls $s$ identities and hires $m$ humans, then
the total number of challenges solved in any window $d$ is at most
$m \cdot \tau_h$.
Maintaining $s$ identities at difficulty $k(d)$ therefore requires
\[
m \;\ge\; \Omega\!\left(\frac{s \cdot k(d)}{\tau_h}\right),
\]
establishing linear Sybil cost in the number of identities.

\paragraph{What kinds of challenges are envisioned?}
HCO is an idealized primitive, but the intended instantiations are closer to
\emph{rate-limited, identity-bound liveness and attention tests} than to a one-shot
account-creation CAPTCHA. Concretely, a challenge can be implemented as an
interactive micro-task that (i) requires a short real-time human response,
(ii) is bound to the validator's key and the current window identifier, and
(iii) admits public verification of a well-formed response transcript.
Examples include time-limited perceptual/matching tasks with randomized prompts,
or lightweight interactive protocols whose randomness is generated at the window
boundary, preventing precomputation. In deployments, such challenges could be
answered through a light client (UI-mediated) and verified via publicly checkable
transcripts bound to the validator key and window identifier.

\paragraph{Why the assumption is separable from consensus mechanics.}
Our protocol and proofs treat HCO as the \emph{source of non-parallelizability}:
any concrete construction is acceptable as long as it preserves the two properties
used in the analysis—identity binding (H2) and a per-human solve-rate bound per window (H4).
We emphasize that PoCmt does not require global uniqueness of humans (unlike
Proof-of-Personhood); it only requires that \emph{each additional active identity}
demands additional contemporaneous human effort, which is the core barrier behind
Theorem~\ref{thm:linear-sybil}.

\paragraph{Minimal HCO interface used by the analysis.}
Formally, our security arguments depend on HCO only through:
(i) non-reusability across identities (H2) under fresh randomness per window, and
(ii) a per-human solve-rate bound within each window (H4).
All other aspects (task format, UX, delivery channel, and verification plumbing)
are orthogonal to the consensus layer as long as these two guarantees hold.
This isolates the consensus proofs from any particular CAPTCHA-like instantiation.

\paragraph{Human--machine gap assumption.}
The security of PoCmt relies on the existence of tasks for which humans retain a
non-negligible advantage over automated solvers within a bounded time window.
As with cryptographic hardness assumptions, this gap is treated as a modeling
assumption rather than a permanent guarantee.
Crucially, PoCmt does not require tasks to remain globally or indefinitely human-only:
it suffices that, for each operational window, there exists a class of challenges
whose automated success probability remains sufficiently low relative to real-time
human performance. We revisit the long-term fragility of this assumption under
advancing automation as an explicit limitation in Section~\ref{sec:discussion}.
\paragraph{Minimal assumptions required for the analysis.}
Importantly, the security proofs of PoCmt rely only on two properties of HCO:
(i) \emph{identity binding}, which prevents reuse or transfer of solutions
across identities, and (ii) a \emph{per-human rate bound} on solvable challenges
per window.
No assumption is made about the semantic content of challenges beyond these
properties. Any construction satisfying these constraints preserves all
theorems in Sections~\ref{sec:cost} and~\ref{sec:security}.

\subsection{Commitment Dynamics}
\label{sec:dynamics}

Commitment evolves deterministically according to the following rules.

\paragraph{Human engagement (windowed boost).}
Let $x_v(d) \in \{0,\dots,k(d)\}$ be the number of challenges solved
by validator $v$ during human window $d$.
At the boundary between windows $d$ and $d{+}1$,
\[
H_v(d{+}1) = H_v(d) + \kappa_h \cdot x_v(d).
\]
Within a window, $H_v(t)$ remains constant.

\paragraph{Availability update.}
At each consensus epoch $t$,
\[
U_v(t{+}1) =
\begin{cases}
	U_v(t) + \kappa_u, & \text{if $v$ is online},\\
	U_v(t)e^{-\lambda}, & \text{if $v$ is offline},
\end{cases}
\]
for parameters $\kappa_u > 0$ and $\lambda > 0$.

\paragraph{Participation and slashing.}
If $v$ follows the protocol in epoch $t$,
\[
P_v(t{+}1) = P_v(t) + \kappa_p.
\]
If $v$ equivocates or violates protocol rules,
\[
P_v(t{+}1) = \delta \, P_v(t),
\quad 0 < \delta < 1.
\]

\paragraph{Score recomputation.}
After applying all updates, the commitment score is recomputed as
\[
CS_v(t{+}1) = \alpha H_v(t{+}1) + \beta P_v(t{+}1) + \gamma U_v(t{+}1).
\]

This dynamics cleanly separates human effort, protocol behavior,
and availability, while producing a single time-varying weight
used by the consensus protocol.

\paragraph{Participation assumptions and incentives (scope).}
Our core theorems follow the standard ``honest-but-participating'' assumption:
honest validators continue to execute the protocol and, when required, provide
the prescribed human engagement. We do not model utilities, fatigue, churn, or
validator replacement in the security proofs, since the cost-theoretic separation
hinges on the \emph{non-parallelizable} nature of human-time rather than on any
particular reward scheme. Incentives and sustained-effort considerations are
therefore treated as orthogonal deployment questions and discussed explicitly as
limitations and design space in Section~\ref{sec:discussion}.

\section{PoCmt Consensus Protocol}
\label{sec:protocol}

We now describe the PoCmt consensus protocol.
PoCmt integrates a human-time–based commitment resource into a
permissionless consensus mechanism while remaining compatible with
classical leader-election, chain-growth, and finality frameworks.
A key design feature is that PoCmt operates over two coupled timescales:
frequent \emph{consensus epochs} and coarser-grained \emph{human windows},
as formalized in Section~\ref{sec:model}.

\subsection{Two-Timescale Execution Model}
\label{sec:timescales}

Time is divided into discrete consensus epochs $t = 0,1,2,\dots$, which
drive leader election, block proposal, and message exchange.
In parallel, time is partitioned into human windows
$d = 0,1,2,\dots$, each spanning multiple epochs (e.g., a day).

The commitment state of validator $v$,
$S_v(t) = (H_v(t), P_v(t), U_v(t))$, evolves on both timescales:
\begin{itemize}
	\item \textbf{Human engagement} $H_v$ is updated only at human-window
	boundaries, based on solutions to Human Challenge Oracle (HCO)
	challenges issued during that window.
	\item \textbf{Participation} $P_v$ and \textbf{availability} $U_v$
	are updated at every epoch, reflecting protocol behavior and online
	presence.
\end{itemize}

Within a given human window, the human-engagement component $H_v$
remains constant. As a result, commitment scores vary smoothly across
epochs, ensuring that human-time contributes as a persistent, slowly
varying resource rather than a per-epoch signal.

\subsection{Protocol Overview}
\label{sec:overview}

PoCmt proceeds through a repeated sequence of window-level and
epoch-level actions.

\paragraph{Window-level actions.}
At the beginning of each human window $d$:
\begin{enumerate}
	\item The HCO issues $k(d)$ fresh identity-bound challenges
	$\{\chi_{v,d,j}\}_{j=1}^{k(d)}$ to each validator $v$.
	\item Validators submit solutions within the real-time response
	window specified by the oracle.
	\item For each validator $v$, the number of successfully solved
	challenges $x_v(d) \in \{0,\dots,k(d)\}$ is recorded.
	\item At the window boundary, the human-engagement state is updated as
	\[
	H_v(d{+}1) = H_v(d) + \kappa_h \cdot x_v(d).
	\]
\end{enumerate}

\paragraph{Epoch-level actions.}
For each consensus epoch $t$ within the current human window:
\begin{enumerate}
	\item Validators exchange consensus messages and participate in block
	proposal and validation.
	\item Participation scores $P_v(t)$ are updated; equivocation or other
	protocol violations trigger slashing.
	\item Online presence is observed; availability scores $U_v(t)$ either
	increase or decay.
	\item Commitment scores
	\[
	CS_v(t) = \alpha H_v(d(t)) + \beta P_v(t) + \gamma U_v(t)
	\]
	are recomputed.
	\item Leader (and optional committee) election is executed using
	commitment-weighted randomness.
\end{enumerate}

This separation ensures that human-time influences consensus weight
persistently over many epochs, while fast-timescale dynamics govern
safety and liveness.

\subsection{Leader Election via Commitment}
\label{sec:leader-election}

PoCmt assumes access to a public randomness source, instantiated via
verifiable random functions (VRFs), threshold randomness, or an external
beacon chain.

Each validator $v$ holds a VRF secret key $\mathsf{sk}_v$ and computes
\[
r_v(t) = \mathsf{VRF}(\mathsf{sk}_v,\; t \parallel \mathsf{rand}(t)).
\]

Validator $v$ is eligible to propose a block in epoch $t$ if
\[
r_v(t) < \tau_v(t),
\qquad
\tau_v(t) = \Theta \cdot
\frac{CS_v(t)}{\sum_{u \in V(t)} CS_u(t)},
\]
where $\Theta$ is chosen such that the expected number of eligible
leaders per epoch is close to one.

Because $H_v$ changes only at human-window boundaries, leader-election
probabilities are piecewise constant within each window. This prevents
rapid, automation-driven amplification of influence and directly
implements commitment-proportional sampling, which underlies the
fairness analysis in Lemma~\ref{lem:fairness}.

\paragraph{Optional committee formation.}
PoCmt may additionally sample a committee for BFT-style consensus using
the same VRF outputs. Each validator is selected independently with
probability
\[
p_v(t) =
c \cdot
\frac{CS_v(t)}{\sum_{u \in V(t)} CS_u(t)},
\]
for a parameter $c$ controlling committee size. This makes PoCmt
compatible with HotStuff- or Algorand-style finality mechanisms.

\subsection{Block Proposal and Validation}
\label{sec:block-proposal}

If $\ell_t$ is selected as leader in epoch $t$, it proposes a block
$B_t$ containing:
\begin{itemize}
	\item parent hash,
	\item epoch number $t$,
	\item transaction set,
	\item VRF proof of eligibility,
	\item leader signature,
	\item optional evidence related to recent commitment updates
	(e.g., window-level HCO solution transcripts).
\end{itemize}

Upon receiving $B_t$, validators perform the following checks:
\begin{enumerate}
	\item verify the VRF proof and leader signature,
	\item check consistency with local fork choice,
	\item validate transactions and detect equivocation,
	\item verify any included commitment-related evidence.
\end{enumerate}

Blocks failing any check are rejected, and valid slashing proofs are
propagated.

\paragraph{Slashing conditions.}
A validator is slashable if it:
\begin{itemize}
	\item produces conflicting blocks for the same epoch,
	\item submits invalid VRF proofs,
	\item forges or misrepresents commitment-related state.
\end{itemize}
Slashing applies a multiplicative penalty
$P_v \leftarrow \delta P_v$ with $0 < \delta < 1$, reducing adversarial
weight without erasing historical human engagement.

\subsection{Fork Choice and Optional Finality}
\label{sec:fork-choice}

PoCmt supports two deployment configurations.

\paragraph{Weighted longest-chain rule.}
Validators adopt the chain maximizing
\[
W(\text{chain}) =
\sum_{B \in \text{chain}} CS_{\text{leader}(B)}(t_B),
\]
where $t_B$ denotes the epoch in which block $B$ was proposed. Under a
weighted honest majority, the honest chain grows faster in expectation,
as required for backbone-style chain growth.

\paragraph{Optional BFT-style finality.}
PoCmt may incorporate a HotStuff-like finality gadget in which validators
vote on blocks with weight proportional to $CS_v(t)$. A block finalizes
once it collects more than $2/3$ of the total commitment weight,
yielding deterministic finality under partial synchrony.

\subsection{Epoch and Window Timeline Summary}
\label{sec:timeline}

The execution of PoCmt can be summarized as follows.

\paragraph{Human window $d$:}
\begin{enumerate}
	\item HCO emits $k(d)$ challenges $\{\chi_{v,d,j}\}$ to each validator.
	\item Validators solve challenges and update $H_v$ at the window boundary.
\end{enumerate}

\paragraph{Each epoch $t$ within window $d$:}
\begin{enumerate}
	\item Consensus messages exchanged; $P_v(t)$ updated.
	\item Online/offline status applied; $U_v(t)$ updated.
	\item Commitment scores $CS_v(t)$ recomputed.
	\item Leader (and optional committee) elected via VRF.
	\item Leader proposes $B_t$; validators verify and update forks.
	\item Slashing proofs propagated and applied.
\end{enumerate}

This completes the description of the PoCmt consensus protocol.

\section{Cost-Theoretic Analysis and Sybil Resistance}
\label{sec:cost}

This section formalizes the cost model underlying PoCmt and establishes its
central guarantee: sustaining $s$ adversarial identities over time requires
\emph{linear human-time effort}, independent of capital, hardware, or network
bandwidth. We further show that this property is unattainable in protocols whose
weight derives solely from \emph{parallelizable} resources such as computation
or stake.

A key distinction from classical models is that human engagement in PoCmt evolves
on the coarser \emph{human-window} timescale (Section~\ref{sec:time-model}),
while consensus execution occurs over fine-grained epochs. Accordingly, we
measure adversarial cost \emph{per window} and translate it into long-run
influence via the commitment score dynamics of Section~\ref{sec:dynamics}.

\subsection{Adversarial Effort Model}
\label{sec:adv-model}

Consider an adversary controlling a set $A$ of $s$ identities. Each adversarial
identity $a\in A$ maintains a commitment state $S_a(t)=(H_a(t),P_a(t),U_a(t))$ as
defined in Section~\ref{sec:model}. Maintaining nontrivial commitment requires
three qualitatively different resources:

\begin{itemize}
	\item \textbf{Human-time cost $c_h$:} solving HCO challenges during each human
	window $d$ in order to increase (or sustain growth of) $H_a$.
	
	\item \textbf{Protocol/availability cost $c_p$:} keeping nodes online and
	participating in messaging to prevent decay of $U_a$ and to accrue $P_a$.
	
	\item \textbf{Slashing risk $c_s$:} expected loss in participation score due to
	equivocation, invalid behavior, or coordination failures.
\end{itemize}

We emphasize that $c_p$ and $c_s$ can be covered by machines and capital, whereas
$c_h$ is irreducibly tied to human cognitive effort.

\paragraph{Window-level accounting.}
Let $k(d)$ be the challenge rate (``difficulty'') in window $d$
(Section~\ref{sec:HCO}). Let $x_a(d)\in\{0,\dots,k(d)\}$ be the number of valid
solutions submitted by identity $a$ during window $d$. By the update rule in
Section~\ref{sec:dynamics}, human engagement increases as
$H_a(d{+}1)=H_a(d)+\kappa_h x_a(d)$.
Thus, any attempt to maintain many identities with growing engagement must
sustain a large aggregate number of valid solutions per window.

\paragraph{Human-time hardness.}
By the HCO assumptions (Section~\ref{sec:HCO}), challenges are identity-bound and
time-limited, and a single human can solve at most $\tau_h=O(1)$ challenges per
window. If the adversary hires $m$ humans, then the total number of challenges
it can solve in window $d$ is at most $m\cdot\tau_h$.

This yields the core constraint behind PoCmt:

\begin{center}
	\emph{Human-time is non-parallelizable beyond the number of humans engaged in
		a window.}
\end{center}

\subsection{Linear Sybil Cost in PoCmt}
\label{sec:linear-cost}

We now formalize the linear human-time barrier and connect it to sustained
adversarial influence under commitment-weighted leader election.

\begin{lemma}[Window-level linear human-time requirement]
	\label{lem:linear-cost}
	Fix a human window $d$ with challenge rate $k(d)$. For any adversary that
	maintains $s$ identities and submits in total $X(d)=\sum_{a\in A} x_a(d)$ valid
	HCO solutions in window $d$, the required human-time satisfies
	\[
	X(d) \;\le\; m\cdot \tau_h,
	\]
	where $m$ is the number of humans available to the adversary. In particular,
	achieving $x_a(d)\ge 1$ for all $s$ identities (nontrivial engagement across
	all identities) requires $m=\Omega(s/\tau_h)$, and achieving
	$x_a(d)=k(d)$ for all identities requires $m=\Omega(sk(d)/\tau_h)$.
\end{lemma}

\begin{proof}[Proof sketch]
	By (H2) identity binding, solutions cannot be reused across identities or
	challenge indices. By (H3) the response window prevents precomputation or
	stockpiling. By (H4), each human can solve at most $\tau_h$ challenges per
	window. Therefore $m$ humans can solve at most $m\tau_h$ challenges in that
	window, implying the stated bounds.
\end{proof}

Lemma~\ref{lem:linear-cost} is purely about \emph{feasible engagement} per window.
We now relate it to adversarial commitment growth relative to honest validators.

\paragraph{Honest growth baseline.}
Under the honest-but-participating assumption, honest validators solve their
challenges and remain online with high probability, so their commitment grows
monotonically: $H$ increases at window boundaries, while $P$ and $U$ typically
increase across epochs (Section~\ref{sec:dynamics}).

\paragraph{Adversarial asymmetry.}
If adversarial identities skip HCO solutions, their $H_a$ stagnates at the
window boundary. If they go offline, $U_a$ decays exponentially; if they
misbehave, $P_a$ is multiplicatively slashed. Thus, without sufficient human-time
capacity, large Sybil sets inevitably contain many identities that fail to
accumulate engagement and cannot sustain high commitment scores.

\begin{lemma}[Sublinear human-time implies bounded adversarial weight]
	\label{lem:adv-vs-honest}
	If an adversary controlling $s$ identities has human-time capacity
	$m=o(s)$ per window (equivalently, $m\tau_h=o(s)$ solved challenges per window),
	then the aggregate adversarial commitment weight cannot track the honest
	aggregate weight asymptotically. Concretely, for sufficiently large time,
	\[
	W_A(t) \ll W_H(t),
	\]
	where $W_A(t)=\sum_{a\in A} CS_a(t)$ and $W_H(t)$ is defined analogously.
\end{lemma}

\begin{proof}[Proof sketch]
	With $m=o(s)$, Lemma~\ref{lem:linear-cost} implies that only a vanishing fraction
	of identities can receive even a constant number of valid solutions per window.
	Those identities have stagnant $H_a$ and, due to availability decay and/or
	slashing risk, do not sustain high $CS_a(t)$ in aggregate. Meanwhile, honest
	identities accumulate engagement and maintain participation/availability, so
	$W_H(t)$ grows at a strictly larger effective rate. This yields a persistent
	gap and therefore $W_A(t)\ll W_H(t)$ asymptotically.
\end{proof}

We now state the main cost-theoretic result.

\begin{theorem}[PoCmt achieves linear Sybil cost]
	\label{thm:linear-sybil}
	Sustaining $s$ adversarial identities with nontrivial commitment over $W$
	human windows requires
	\[
	\Omega(sW)
	\]
	units of human-time effort (equivalently, $\Omega(s)$ per window). Consequently:
	(i) PoCmt is asymptotically Sybil-hard; (ii) capital, hardware, and bandwidth
	cannot substitute for human engagement; and (iii) identity-amplification
	strategies effective in PoW/PoS do not apply.
\end{theorem}

\begin{proof}[Proof sketch]
	By Lemma~\ref{lem:linear-cost}, maintaining engagement across $s$ identities
	requires $\Omega(s)$ solved challenges per window and hence $\Omega(s)$ human-time
	per window. Lemma~\ref{lem:adv-vs-honest} shows that sublinear capacity causes
	adversarial weight to fall behind honest weight, preventing sustained influence.
	Summing over $W$ windows yields $\Omega(sW)$ total human-time.
\end{proof}

\paragraph{Remarks: outsourcing and coordination.}
PoCmt does not prevent an adversary from outsourcing HCO solutions by hiring many
humans, nor does it enforce one-person--one-identity. These behaviors are allowed
by the model and do not contradict Theorem~\ref{thm:linear-sybil}. Rather, they
clarify the security regime: PoCmt shifts adversarial scaling from
machine-parallelizable resources (capital/hardware reuse) to explicit labor
capacity. Coordination can aggregate human effort, but cannot compress or
multiply it across identities; influence remains proportional to contemporaneous
human-time per window.

\subsection{Impossibility of Linear Sybil Cost in Parallelizable-Resource Protocols}
\label{sec:impossibility}

We contrast PoCmt with protocols whose influence derives from a
\emph{parallelizable resource}, such as computational power or stake.

\paragraph{Resource-parallelizable protocols.}
Let $R$ denote the adversary’s total resource. Such protocols satisfy:
\begin{enumerate}
	\item \textbf{Additivity:} influence depends only on the total resource
	$R_{\mathrm{tot}}=\sum_i R_i$, not on identity count.
	\item \textbf{Divisibility:} $R$ can be split arbitrarily across identities at
	negligible cost.
	\item \textbf{Monotonicity:} influence is non-decreasing in $R$.
\end{enumerate}
These assumptions capture PoW, PoS, and related designs.

\begin{lemma}[Zero marginal Sybil cost]
	\label{lem:zero-marginal}
	In any resource-parallelizable protocol, an adversary with resource $R$ can
	create $s$ identities with the \emph{same total influence} at cost $C(R)$,
	independent of $s$. Thus the marginal Sybil cost is asymptotically zero.
\end{lemma}

\begin{proof}
	By divisibility, split $R$ across $s$ identities. By additivity, total influence
	remains unchanged. Identity creation is cheap relative to acquiring $R$, hence
	the cost depends only on $R$ and not on $s$.
\end{proof}

\begin{corollary}[Impossibility of linear Sybil cost in PoW/PoS]
	\label{cor:impossibility}
	Any protocol based solely on parallelizable resources cannot enforce linear
	Sybil cost. In particular, PoW and PoS admit asymptotically zero marginal cost
	for adding identities.
\end{corollary}

This establishes a sharp cost-theoretic separation: PoCmt lies outside the class
of parallelizable-resource protocols and achieves a property provably unattainable
in PoW/PoS-style systems.

\section{Security Analysis of PoCmt}
\label{sec:security}

This section argues that PoCmt satisfies backbone-style security properties
under the model of Section~\ref{sec:model} and the human-time constraints of
Section~\ref{sec:HCO}. A key subtlety is the two-timescale structure:
human engagement $H_v$ is updated on the coarse \emph{human-window} timescale,
while consensus progresses over fine-grained epochs. Accordingly, our arguments
combine (i) per-epoch leader-election probabilities induced by the current
commitment weights, with (ii) window-scale \emph{drift} of total honest versus
adversarial weight.

\paragraph{Assumptions, scope, and robustness.}
PoCmt operates under the standard partial synchrony assumption
(Section~\ref{sec:protocol}) and requires validators to maintain reasonably
stable connectivity to preserve availability score $U_v(t)$.
Validators with persistently poor network access may experience availability
decay despite honest intent. This limitation is not unique to PoCmt:
leader-based and committee-based protocols fundamentally rely on timely message
exchange after GST. Importantly, availability in PoCmt is a \emph{soft weighting
	signal} rather than a correctness condition: temporary unavailability reduces
future influence but does not invalidate past blocks, trigger slashing, or
create safety violations. Thus, connectivity primarily affects \emph{relative
	weight} and \emph{liveness bounds}, not protocol correctness.

Finally, our results rely on the Human Challenge Oracle assumption that there
exist time-bounded tasks for which humans retain an advantage over automated
solvers (Section~\ref{sec:HCO}). As in cryptographic hardness assumptions, this
gap is treated as a modeling primitive: PoCmt requires a practical human--machine
gap over bounded operational horizons, not an absolute or permanent guarantee.

\subsection{Weighted Honest Majority and Drift Invariant}
\label{sec:drift}

Let $CS_v(t)$ be the commitment score used for leader election at epoch $t$ and
define total honest and adversarial weight:
\[
W_H(t)=\sum_{v\in H} CS_v(t),
\qquad
W_A(t)=\sum_{a\in A} CS_a(t),
\qquad
W_{\mathrm{tot}}(t)=W_H(t)+W_A(t).
\]

\begin{definition}[Weighted honest majority]
	\label{def:whm}
	PoCmt satisfies a weighted honest majority if there exists $\rho<1/2$ such that
	for all epochs $t$,
	\[
	W_A(t) \le \rho \cdot W_{\mathrm{tot}}(t).
	\]
\end{definition}

Definition~\ref{def:whm} is the standard condition under which commitment-weighted
leader election favors honest participants. The central reason this condition is
stable in PoCmt is cost-theoretic: adversarial weight cannot scale without linear
human-time per window (Theorem~\ref{thm:linear-sybil}), whereas honest validators
continue to accrue engagement, participation, and availability under normal
operation (Section~\ref{sec:dynamics}).

\begin{lemma}[Commitment drift invariant (window-scale)]
	\label{lem:drift}
	Assume that after GST honest validators (i) solve HCO challenges when issued
	and (ii) remain online except for transient faults. If the adversary invests
	sublinear human-time capacity $m=o(s)$ per window (equivalently, solves
	$o(s)$ identity-bound challenges per window), then across window boundaries
	the honest advantage does not decrease: for any window boundary times $t<t'$,
	\[
	W_H(t') - W_A(t') \;\ge\; W_H(t) - W_A(t).
	\]
\end{lemma}

\begin{proof}[Proof sketch]
	Over each window, honest validators obtain a net positive engagement increment
	through $H_v$ updates and continue accruing $P_v$ and $U_v$ at the epoch scale.
	The adversary can increase engagement only through identity-bound HCO solutions.
	Under sublinear capacity, Lemma~\ref{lem:linear-cost} implies that a growing
	fraction of adversarial identities cannot receive solutions, so their engagement
	stagnates and their availability decays when rotated offline. Any equivocation
	additionally triggers multiplicative slashing of $P_v$, further reducing future
	weight. Thus the honest-minus-adversarial gap is non-decreasing across windows.
\end{proof}

\subsection{Safety}
\label{sec:safety}

We present safety for (i) commitment-weighted longest-chain selection and (ii)
optional weighted BFT finality. For the longest-chain mode, we phrase the
guarantee in terms of the common-prefix property.

Let the per-epoch probability that the elected leader is honest be
\[
p_H(t)=\frac{W_H(t)}{W_{\mathrm{tot}}(t)}.
\]
Under Definition~\ref{def:whm}, $p_H(t) > 1/2$ for all epochs.

\begin{theorem}[Common-prefix (weighted longest-chain)]
	\label{thm:cp}
	Assume partial synchrony after GST and weighted honest majority:
	$W_A(t) < \tfrac{1}{2}W_{\mathrm{tot}}(t)$ for all $t$.
	Then for any confirmation depth $k$, the probability that two honest validators
	adopt chains that disagree on a block that is $k$-deep is negligible in $k$.
\end{theorem}

\begin{proof}[Proof sketch]
	This follows the standard backbone intuition adapted to commitment-weighted leader
	election: when $p_H(t)>1/2$, the honest chain grows faster in expectation than any
	adversarial fork. A $k$-deep disagreement requires the adversary to sustain a
	competing fork over many decisive epochs, which entails an atypically long run of
	adversarial leader elections. The probability of such an event decays exponentially
	in $k$. Moreover, equivocation is slashable, reducing adversarial future weight and
	therefore strengthening the drift effect over time (Lemma~\ref{lem:drift}).
\end{proof}

\paragraph{Optional weighted BFT finality.}
If PoCmt is paired with a HotStuff-style finality gadget where votes are weighted
by $CS_v(t)$, safety becomes deterministic under the usual quorum intersection
condition.

\begin{theorem}[Safety with weighted BFT finality]
	\label{thm:safety-bft}
	If a block is finalized only after collecting $>2/3$ of total commitment weight
	in votes, and $W_A(t) < \tfrac{1}{3}W_{\mathrm{tot}}(t)$ for all $t$, then two
	conflicting finalized blocks cannot exist.
\end{theorem}

\begin{proof}[Proof sketch]
	Any two quorums each exceeding $2/3$ of the total weight intersect in more than
	$1/3$ of the total weight. Since the adversary controls strictly less than $1/3$,
	at least one honest validator would need to vote for both conflicting blocks,
	contradicting deterministic honest behavior.
\end{proof}

\subsection{Liveness and Expected Leader Delay}
\label{sec:liveness}

\begin{theorem}[Liveness]
	\label{thm:liveness}
	Assume partial synchrony with unknown GST and that at least one honest validator
	remains online after GST. Then the chain grows without bound and honest blocks
	are eventually confirmed (or finalized, if the BFT gadget is enabled).
\end{theorem}

\begin{proof}[Proof sketch]
	After GST, honest messages propagate within bounded delay. Under the drift
	invariant (Lemma~\ref{lem:drift}), honest weight does not lose ground to the
	adversary across windows, and in typical executions it increases relative to
	adversarial weight. Hence there exists $p_{\min}>0$ such that $p_H(t)\ge p_{\min}$
	for all sufficiently large $t$. Therefore, an honest leader appears within any
	$k$ consecutive epochs with probability $1-(1-p_{\min})^k$, which converges to $1$
	exponentially fast. Each honest-led proposal extends the chain and is adopted by
	honest validators under fork choice, implying unbounded growth and eventual
	confirmation/finality.
\end{proof}

\paragraph{Expected delay to an honest leader.}
If $p_H(t)\ge p_{\min}$, then the waiting time (in epochs) to the next honest
leader is geometric with
\[
\mathbb{E}[T_{\mathrm{honest}}] = \frac{1}{p_{\min}}.
\]

\subsection{Long-Range and Zombie-Validator Robustness}
\label{sec:longrange}

Long-range attacks and ``zombie'' validators exploit identities that were once
powerful but later became inactive. PoCmt mitigates this via time-dependent
scoring together with availability decay and (optional) protocol slashing.

\begin{lemma}[Inactivity implies loss of effective influence]
	\label{lem:long-range}
	If a validator is offline for $k$ consecutive epochs, then its availability
	decays exponentially:
	\[
	U_v(t{+}k) \le U_v(t)e^{-\lambda k},
	\]
	and hence its availability contribution to the commitment score satisfies
	$\gamma U_v(t{+}k)\le \gamma U_v(t)e^{-\lambda k}$.
\end{lemma}

\paragraph{Discussion (zombie resistance).}
Lemma~\ref{lem:long-range} formalizes that inactive identities cannot retain
availability-derived influence and therefore cannot remain competitive indefinitely
without returning online. In addition, inactive validators do not accrue
participation $P_v$ and may lose it through slashing if they attempt equivocation.
If a deployment desires stronger long-range suppression, PoCmt can optionally add
explicit aging/decay to stale engagement $H_v$ (a parameter choice orthogonal to
our cost-theoretic separation), which only strengthens the conclusion.

\subsection{Fairness of Leader Election}
\label{sec:fairness}

PoCmt samples leaders proportionally to commitment score, yielding long-run
fairness among honest validators.

\begin{lemma}[Leader-election fairness]
	\label{lem:fairness}
	If two honest validators have asymptotically identical commitment-score
	trajectories, then the fraction of epochs in which each is elected leader
	converges to their commitment-proportional probability:
	\[
	\lim_{T\to\infty}
	\frac{\#\{\text{epochs where } v \text{ leads}\}}{T}
	=
	\lim_{T\to\infty}
	\frac{CS_v(T)}{\sum_u CS_u(T)}.
	\]
\end{lemma}

\begin{proof}[Proof sketch]
	Leader election is implemented via independent VRF trials with success
	probability proportional to $CS_v(t)$ (Section~\ref{sec:leader-election}).
	Standard concentration (law of large numbers) implies empirical leader
	frequencies converge to these probabilities over long horizons.
\end{proof}

\paragraph{Remark (outsourcing and coordinated humans).}
As discussed in Section~\ref{sec:cost}, PoCmt does not preclude an adversary from
hiring many humans or coordinating human groups. These behaviors do not violate
safety or liveness; they simply increase $W_A(t)$ by increasing the available
human-time capacity. The security statements above are therefore best interpreted
as \emph{weight-relative}: as long as a weighted honest majority holds (Definition~\ref{def:whm}),
PoCmt enjoys the same backbone-style guarantees as other leader-based protocols,
while differing fundamentally in how adversarial weight can be scaled.


\section{Simulation Framework and Evaluation}
\label{sec:simulation}

We complement our analytical results with a lightweight yet expressive
simulation of PoCmt. The purpose of the simulation is not performance
benchmarking, but \emph{mechanism validation}: we empirically demonstrate that
(i) commitment exhibits a drift advantage toward honest validators,
(ii) adversarial influence is fundamentally limited by \emph{human-time capacity}
rather than capital or identity count, and (iii) commitment-weighted leader
election is fair in the long run.

Concretely, the experiments are designed to mirror three analytical claims:
(i) the \emph{drift invariant} (Lemma~\ref{lem:drift}),
(ii) the \emph{linear Sybil-cost} barrier induced by human-time
(Theorem~\ref{thm:linear-sybil}), and
(iii) \emph{leader-election fairness} under commitment-proportional sampling
(Lemma~\ref{lem:fairness}).

The simulator is intentionally minimalistic and self-contained. All design
choices correspond directly to the formal model of Section~\ref{sec:model} and
the Human Challenge Oracle assumptions of Section~\ref{sec:HCO}, allowing clear
interpretation of outcomes.

\subsection{Simulation Model and Two-Timescale Time Structure}

For simplicity, each simulation epoch corresponds to a single human window.
This preserves the rate-limiting effect of human engagement while avoiding
an additional nested time loop. All analytical guarantees are window-based
and therefore directly reflected by the simulated epochs.

The simulation follows the two-timescale structure of PoCmt
(Section~\ref{sec:time-model}). Time is discretized into consensus epochs
$t=0,1,2,\dots$, which drive leader election and block production. In parallel,
epochs are grouped into coarser \emph{human windows} indexed by $d=0,1,2,\dots$.
Within each window, validators may solve a small number of HCO challenges, and
the engagement component $H_v$ is updated at the window boundary.

Each validator $v$ maintains a commitment state
\[
S_v(t)=(H_v(t),P_v(t),U_v(t)),
\]
capturing cumulative human engagement, protocol participation, and availability.
The commitment score used for leader election is
\[
CS_v(t)=\alpha H_v(t)+\beta P_v(t)+\gamma U_v(t),
\]
with fixed weights $(\alpha,\beta,\gamma)$.

We initialize $|H|$ honest validators and $s$ adversarial Sybil identities, all
starting from zero commitment. Honest validators solve their HCO challenges with
high probability, remain online with high probability, and follow the protocol.
Adversarial identities are controlled by a centralized adversary subject to an
explicit human-time capacity constraint.

\subsection{Adversarial Constraints and Human-Time Capacity}

The key modeling choice is that adversarial effort is limited by \emph{human-time
	capacity} rather than by capital or identity count. We model a parameter $m$,
representing the maximum number of HCO challenges the adversary can solve per
\emph{human window}. This captures the HCO non-parallelizability assumption:
even if the adversary controls many identities, the only way to increase
aggregate engagement is to increase the number of human-solves per window.

The adversary may allocate these $m$ solved challenges across Sybil identities
(e.g., concentrate on a few, or rotate among many), but cannot exceed $m$ total
solutions per window. This corresponds directly to the linear-cost bound in
Section~\ref{sec:linear-cost} and the asymptotic requirement in
Theorem~\ref{thm:linear-sybil}.

\subsection{Parameters}

Table~\ref{tab:sim-params} summarizes the default simulation parameters. Values
are chosen to illustrate qualitative effects rather than to optimize throughput
or latency.

\begin{table}[t]
	\centering
	\small
	\begin{tabular}{l c}
		\hline
		\textbf{Parameter} & \textbf{Value} \\
		\hline
		Honest validators $|H|$ & $50$ \\
		Sybil identities $s$ & $100$ \\
		Epoch horizon $T$ & $3000$ epochs \\
		Commitment weights $(\alpha,\beta,\gamma)$ & $(1.0,\,0.5,\,0.1)$ \\
		Human engagement boost $\kappa_h$ & $1.0$ (per solved challenge) \\
		Participation boost $\kappa_p$ & $0.5$ (per epoch) \\
		Availability boost $\kappa_u$ & $0.2$ (per online epoch) \\
		Availability decay rate $\lambda$ & $0.05$ (per offline epoch) \\
		Slashing multiplier $\delta$ & $0.1$ \\
		Honest solve probability & $0.98$ \\
		Honest online probability & $0.995$ \\
		Adversarial human-time capacity $m$ & varies (per window) \\
		\hline
	\end{tabular}
	\caption{Simulation parameters. Values are illustrative and chosen to highlight
		qualitative dynamics of PoCmt rather than performance tuning.}
	\label{tab:sim-params}
\end{table}

\subsection{Metrics}

From the simulation trace we derive the following metrics:

\begin{itemize}
	\item \textbf{Commitment trajectories:}
	total honest and adversarial commitment weights
	$W_H(t)=\sum_{v\in H} CS_v(t)$ and $W_A(t)=\sum_{a\in A} CS_a(t)$.
	
	\item \textbf{Leader share:}
	the fraction of epochs in which the elected leader is adversarial.
	
	\item \textbf{Commitment weight share:}
	the ratio $W_A(T)/(W_H(T)+W_A(T))$ at the end of the horizon.
	
	\item \textbf{Fairness:}
	the deviation between empirical leader frequencies and ideal
	commitment-proportional probabilities for honest validators.
\end{itemize}

These metrics correspond directly to the analytical claims of
Sections~\ref{sec:security} and~\ref{sec:cost}.

\subsection{Results}

\paragraph{Commitment drift.}
Figure~\ref{fig:commitment-trajectories} plots total honest and adversarial
commitment over time. Honest commitment grows steadily with a larger slope,
while adversarial commitment grows more slowly because engagement can only be
increased via a bounded number of solved HCO challenges per window.
The gap $W_H(t)-W_A(t)$ is non-decreasing, empirically validating the drift
invariant (Lemma~\ref{lem:drift}).

\begin{figure}[t]
	\centering
	\includegraphics[width=0.85\linewidth]{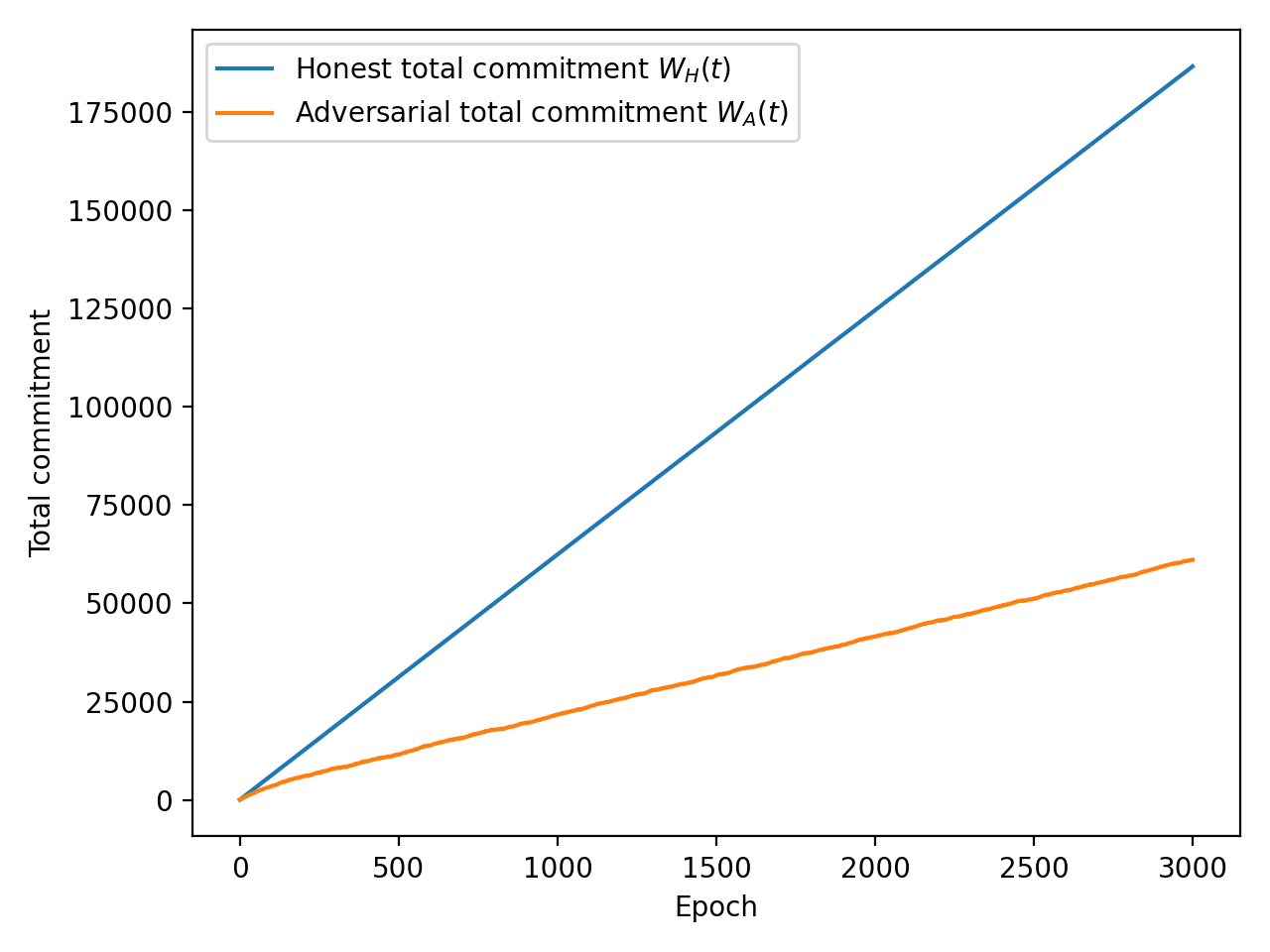}
	\caption{\textbf{Commitment drift under human-time scarcity.}
		Total honest and adversarial commitment weights,
		$W_H(t)=\sum_{v\in H}CS_v(t)$ and $W_A(t)=\sum_{a\in A}CS_a(t)$, over time.
		Honest commitment grows approximately linearly since honest validators solve
		challenges and remain online with high probability, while the adversary is
		constrained by a fixed human-time capacity $m$ (maximum number of solved HCO
		challenges per window). The persistent gap $W_H(t)-W_A(t)$ empirically
		illustrates the drift behavior predicted by Lemma~\ref{lem:drift}.}
	\label{fig:commitment-trajectories}
\end{figure}

\paragraph{Human-time capacity sweep.}
Figures~\ref{fig:capacity-leader} and~\ref{fig:capacity-weight} vary the
adversary’s human-time capacity $m$ while keeping all other parameters fixed.
As $m$ increases, adversarial influence grows smoothly; however, achieving
majority influence requires $m$ itself to be large. Increasing the number of
Sybil identities or reallocating effort alone does not bypass this constraint.

The capacity sweep isolates the core non-parallelizability mechanism: even if
the adversary controls many identities, the only way to increase aggregate
engagement is to increase the number of human-solves per window. This corresponds
directly to the bound in Section~\ref{sec:linear-cost} and the asymptotic
requirement in Theorem~\ref{thm:linear-sybil}.

\begin{figure}[t]
	\centering
	\includegraphics[width=0.85\linewidth]{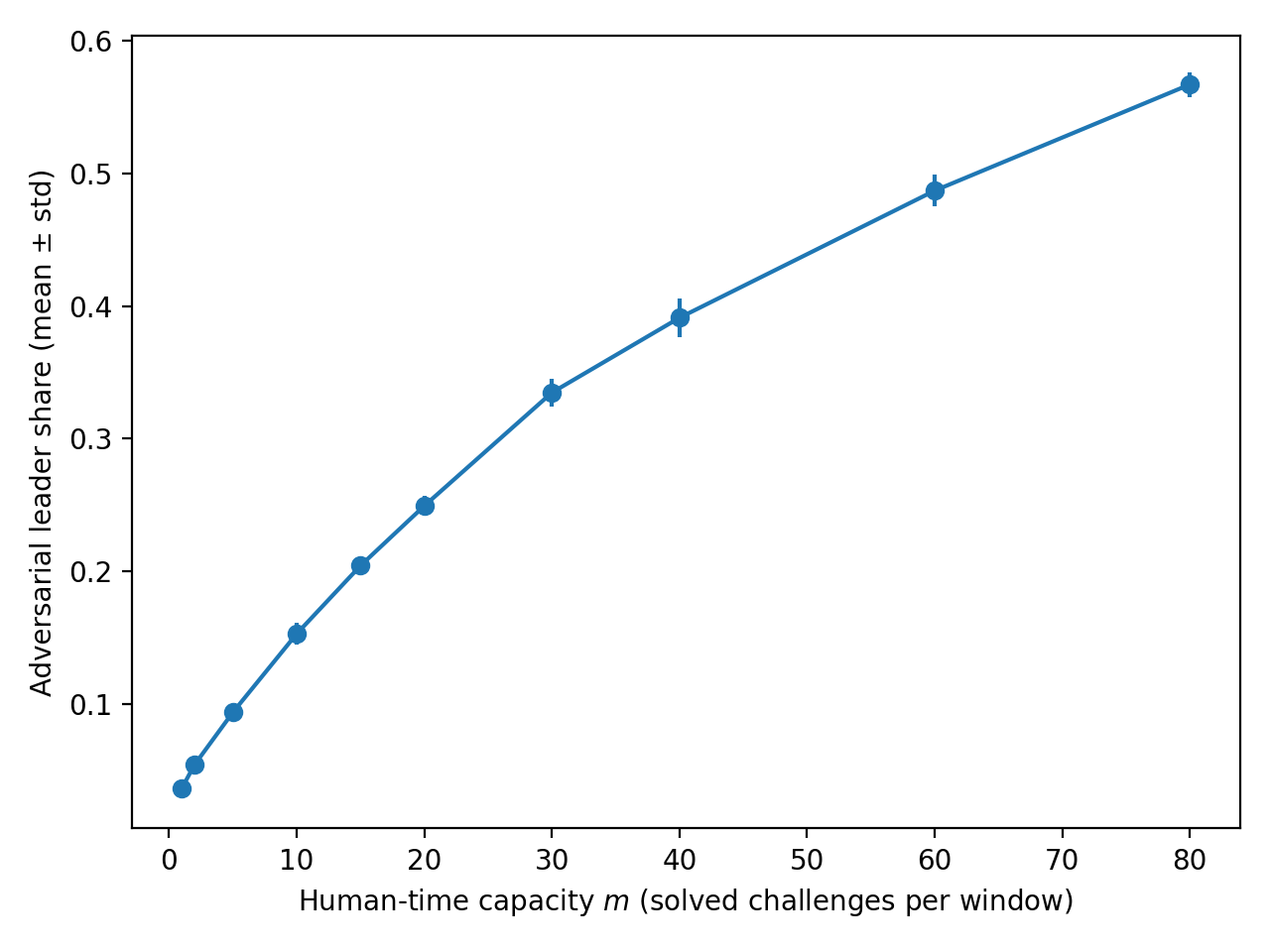}
	\caption{\textbf{Adversarial leader share is controlled by human-time capacity.}
		Mean adversarial leader fraction (with standard deviation across seeds) as a
		function of adversarial human-time capacity $m$ (solved challenges per window).
		Increasing $m$ increases adversarial leader share smoothly; however, achieving
		majority leadership requires $m$ itself to be large. This operationalizes the
		non-parallelizable resource claim: capital or identity replication alone cannot
		push leader share beyond what human-time permits (Theorem~\ref{thm:linear-sybil}).}
	\label{fig:capacity-leader}
\end{figure}

\begin{figure}[t]
	\centering
	\includegraphics[width=0.85\linewidth]{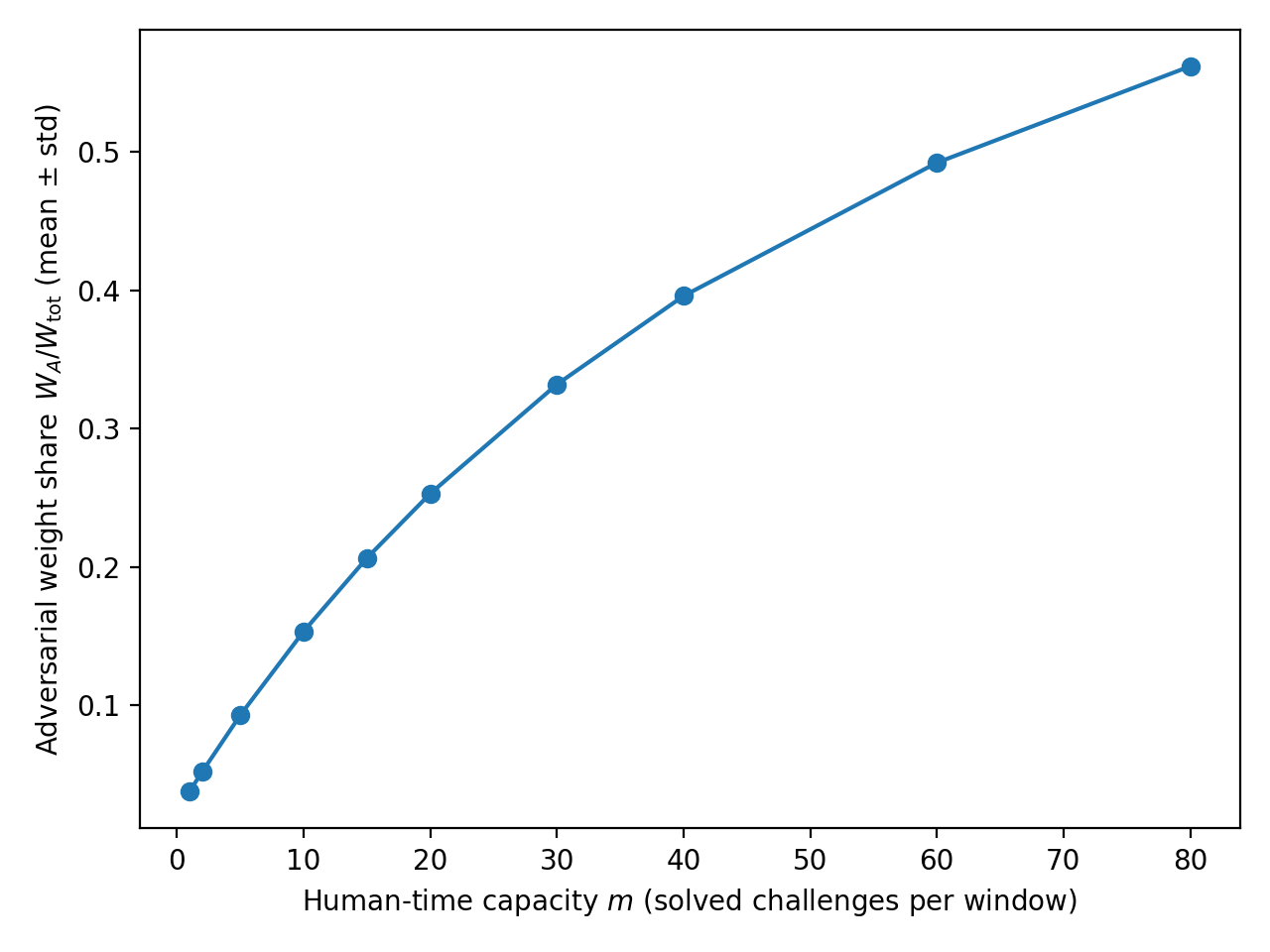}
	\caption{\textbf{Weight share tracks leader share under commitment-proportional sampling.}
		Final adversarial commitment share $W_A(T)/(W_H(T)+W_A(T))$ versus human-time
		capacity $m$. The close agreement between weight share (this figure) and leader
		share (Figure~\ref{fig:capacity-leader}) confirms that the implemented leader
		sampling matches the ideal commitment-proportional rule in
		Section~\ref{sec:leader-election}, and supports the fairness statement of
		Lemma~\ref{lem:fairness}.}
	\label{fig:capacity-weight}
\end{figure}

\paragraph{Fairness of leader election.}
Figure~\ref{fig:fairness} compares empirical leader frequencies of honest
validators against their ideal commitment-proportional probabilities.
We evaluate fairness at the implementation level by comparing empirical leader
frequencies against the ideal commitment-proportional probabilities used by the
protocol (Section~\ref{sec:leader-election}). The resulting alignment provides an
empirical sanity check for Lemma~\ref{lem:fairness}.

\begin{figure}[t]
	\centering
	\includegraphics[width=0.78\linewidth]{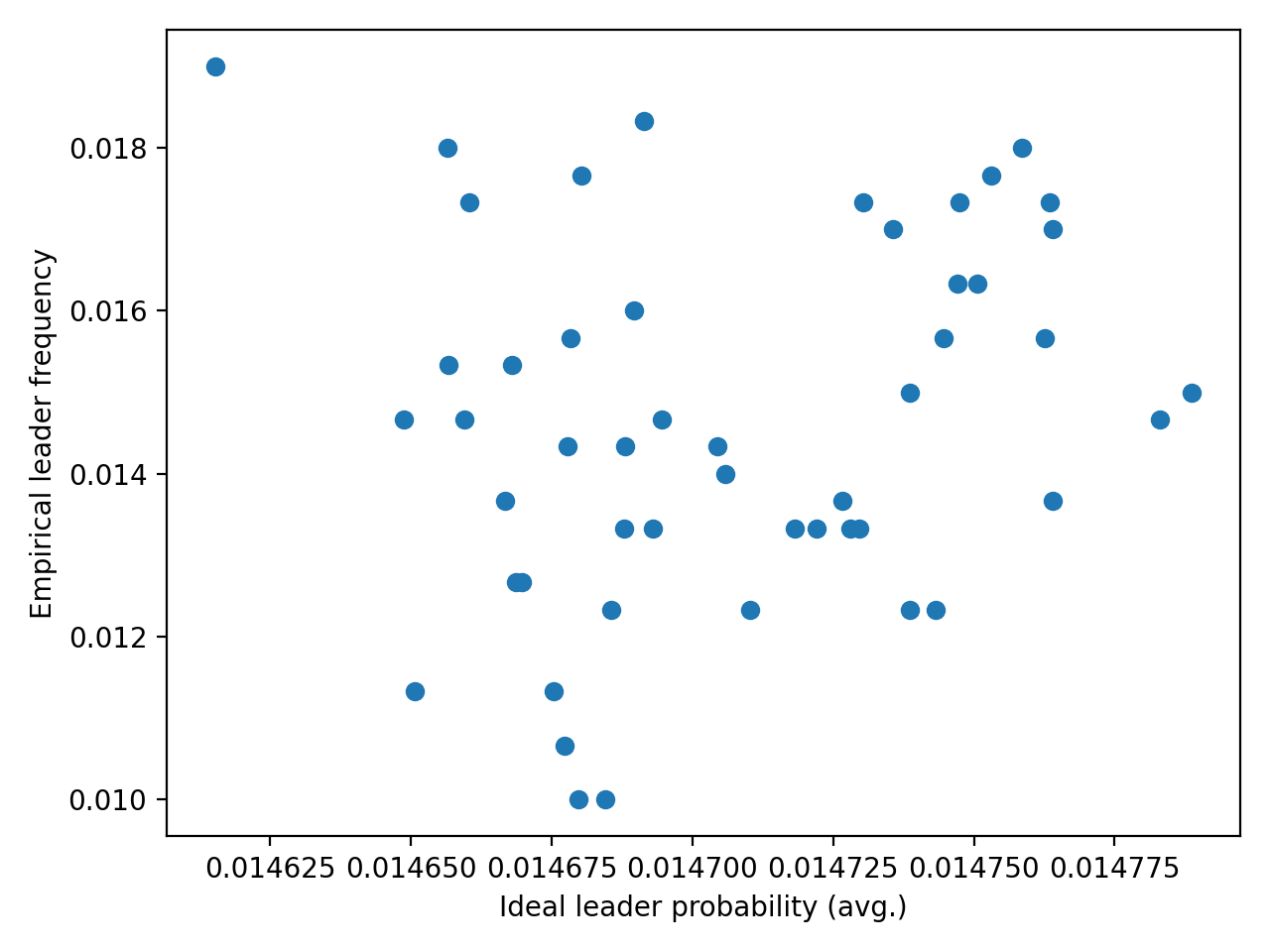}
	\caption{\textbf{Fairness of commitment-weighted leader election among honest validators.}
		For each honest validator, the x-axis reports its ideal leader probability
		averaged over time (proportional to $CS_v(t)$), while the y-axis shows the
		empirical leader frequency over the simulation horizon. Points cluster near the
		diagonal, indicating that leader election does not systematically bias honest
		nodes beyond differences in accumulated commitment, consistent with
		Lemma~\ref{lem:fairness}.}
	\label{fig:fairness}
\end{figure}

\paragraph{Optional ablation: availability decay.}
To stress-test the role of availability in suppressing dormant identities, we
repeat the baseline experiment while varying the decay rate $\lambda$.
Larger $\lambda$ accelerates the loss of $U_v(t)$ for offline identities,
reducing the viability of rotating or intermittently online Sybils. This
experiment complements the long-range and zombie-validator intuition in
Lemma~\ref{lem:long-range} by demonstrating that the availability component acts
as a tunable ``freshness'' filter in practice.

\subsection{Discussion}

Across experiments, the simulation supports the theoretical claims of PoCmt:

\begin{itemize}
	\item Adversarial influence scales with human-time capacity rather than with
	identity count or capital, matching the cost separation of Section~\ref{sec:cost}.
	\item Sustaining many influential Sybil identities requires linear human effort
	per window (Theorem~\ref{thm:linear-sybil}), and sublinear effort leads to drift
	in favor of honest validators (Lemma~\ref{lem:drift}).
	\item Commitment-weighted leader election is empirically fair among honest
	validators (Lemma~\ref{lem:fairness}), consistent with the protocol’s
	commitment-proportional sampling rule (Section~\ref{sec:leader-election}).
	\item Availability decay provides a practical lever for suppressing dormant or
	rotating identities, in line with the long-range robustness intuition
	(Lemma~\ref{lem:long-range}).
\end{itemize}

Together with the analytical results, these experiments provide empirical
evidence that PoCmt realizes a qualitatively different security regime grounded
in non-parallelizable human effort.

\section{Discussion and Limitations}
\label{sec:discussion}

PoCmt introduces a fundamentally human-centric weighting resource for
permissionless consensus, replacing machine-parallelizable resources with
\emph{real-time human effort}. While the preceding sections establish
cost-theoretic guarantees (Section~\ref{sec:cost}), backbone-style security
properties (Section~\ref{sec:security}), and empirical mechanism validation
(Section~\ref{sec:simulation}), several practical considerations arise when
translating PoCmt into deployed systems.

\subsection{Human Effort as a Security Resource}

PoCmt treats human cognition, perception, and temporal availability as scarce
resources. This creates design constraints that are absent in PoW/PoS-like
systems.

\paragraph{Cognitive load and sustainability.}
Validators must periodically solve Human Challenge Oracle (HCO) tasks within
bounded response windows. These tasks must be short enough to avoid fatigue and
attrition, while still maintaining a meaningful human--machine gap. Determining
a stable difficulty--frequency region---and how it interacts with validator
churn---is a core deployment question.

\paragraph{Inclusiveness and accessibility.}
Challenge families must be robust across languages, cultures, and accessibility
conditions. Unlike PoW/PoS, PoCmt’s security depends on human capabilities and
user experience; usability therefore becomes a first-class security parameter.

\subsection{Incentives for Sustained Human Effort}
\label{subsec:incentives}

\noindent\textbf{Why this matters.}
A practical deployment must explain why rational participants continue providing
human effort over long periods, despite fatigue, abandonment, churn, and
replacement. Our analysis focuses on the \emph{security consequences} of using
human-time as the scarce weighting resource and therefore adopts the standard
\emph{honest-but-participating} assumption: honest validators continue solving
HCO challenges and participating over time.

A complete incentive model (utility, fatigue, entry/exit, and labor-market
effects) is orthogonal to the cost-theoretic separation and backbone-style
guarantees, but essential for predicting real-world participation.

\paragraph{Natural incentive hooks.}
PoCmt admits several incentive mechanisms that can sustain effort without
changing the security core:
(i) window-level rewards can be tied directly to recent engagement (e.g., a
challenge-solve reward and/or a multiplicative reward factor depending on recent
$H_v$);
(ii) commitment can act as a reputation-like signal that influences block reward
allocation and committee selection, creating long-run benefits from continued
participation; and
(iii) validator operation can be decomposed into infrastructure operators (who
maintain online nodes) and human solvers (who provide periodic HCO responses),
supporting churn and replacement while preserving the linear-cost barrier.

Formalizing these incentives requires equilibrium analysis under human fatigue,
outsourcing, and churn. We view this as a primary direction for future work, and
note that our security results remain valid under any incentive mechanism that
preserves the HCO non-parallelizability bound per window (Section~\ref{sec:HCO}).

\subsection{Parameter Sensitivity and Tuning}

PoCmt introduces parameters such as engagement boost $\kappa_h$, availability
decay rate $\lambda$, slashing factor $\delta$, and scoring weights
$(\alpha,\beta,\gamma)$. Poor parameterization may reduce either security or
participation:

\begin{itemize}
	\item excessive decay may penalize honest validators with unstable connectivity,
	\item excessive engagement boost may create high-variance growth dynamics,
	\item aggressive slashing may deter participation despite honest intent.
\end{itemize}

Availability sensitivity is explicitly parameterized via $\lambda$ and the
duration of human windows (Section~\ref{sec:time-model}). Larger windows and
smaller decay rates reduce penalties for intermittent connectivity, while smaller
windows and larger decay favor stricter freshness guarantees. This exposes a
deployment-level trade-off between tolerance to unstable connectivity and
responsiveness against dormant or abandoned validators (cf.\ the decay ablation
in Section~\ref{sec:simulation}).

A full characterization of safe and robust parameter regions remains an
important direction for future work.

\subsection{Adversaries Beyond Sybil Flooding}

PoCmt guarantees linear Sybil cost, but several adversarial behaviors remain
relevant.

\paragraph{Human labor markets and outsourcing.}
A well-funded adversary may outsource HCO tasks to hired workers. This does not
violate the security model: outsourcing preserves the linear cost structure of
Theorem~\ref{thm:linear-sybil}. However, it shifts the locus of power from
technical resources to labor markets. In this sense, PoCmt is Sybil-resistant
but not necessarily resistant to economic concentration: an adversary can grow
influence by paying ongoing, observable labor costs.

\paragraph{Advances in automation and the long-run human--machine gap.}
PoCmt relies on the existence of challenge families that preserve a practical
human advantage within bounded time windows (Section~\ref{sec:HCO}). This
assumption is analogous in spirit to standard hardness assumptions, but its
long-run stability may be affected by progress in automation. Sustaining PoCmt
over long time horizons therefore motivates adaptive challenge generation and
diversity: deployments should anticipate periodic updates to challenge families
and difficulty parameters, and should treat the challenge layer as evolvable
infrastructure rather than a fixed primitive.

\paragraph{Coordinated groups of humans.}
PoCmt intentionally does not enforce one-person--one-identity. Coordinated groups
of real humans can accumulate commitment without violating any rule. This does
not undermine the linear Sybil-cost guarantee---each additional unit of influence
still requires additional contemporaneous human effort per window---but it raises
a separate question: what notion of \emph{social} decentralization is achieved?
PoCmt provides decentralization at the level of effort rather than social
structure: it guarantees that influence reflects real-time human participation,
not that humans are uncoordinated.

\subsection{Deployment Constraints}

Practical deployments face logistical and systems-level constraints.

\paragraph{Challenge delivery and privacy.}
Challenges must be delivered securely without enabling precomputation, identity
linking, or leakage of validator metadata. Designs must balance transparency
(public verifiability) against privacy (limiting linkability across windows).

\paragraph{Client environments.}
PoCmt does not require trusted hardware, but real deployments may benefit from
secure client environments (e.g., hardened mobile clients or enclaves) to reduce
spoofing and replay risks and to support reliable challenge interaction.

\paragraph{Connectivity assumptions.}
Because availability contributes to commitment, regions with unstable Internet
access may experience unfair penalties. Adjusting $\lambda$, lengthening human
windows, or introducing grace mechanisms can mitigate this effect, but may trade
off responsiveness against dormant identities.

\subsection{Relation to Proof-of-Personhood}

PoCmt sits between anonymous permissionless systems and global identity-based
approaches. It does not require unique human identities, but:
\begin{itemize}
	\item maintaining many actively engaged identities remains costly but not impossible,
	\item identity resale or delegation cannot be completely prevented,
	\item PoCmt does not provide one-person--one-identity guarantees.
\end{itemize}
Thus PoCmt avoids the governance complexity of global identity systems while
still achieving linear Sybil-cost guarantees. It should be viewed as a consensus
weighting primitive grounded in human-time, not as an identity system.

\subsection{Theoretical Limitations}

Several limitations persist even under idealized assumptions.

\paragraph{Dependence on partial synchrony.}
Human-time engagement and timely HCO responses require weak timing assumptions.
In fully asynchronous settings, real-time commitment becomes ill-defined.

\paragraph{Network partitions.}
Extended partitions may cause honest validators to accumulate availability decay
if they cannot reach the challenge oracle or other validators, temporarily
reducing their commitment and influence.

\paragraph{Throughput constraints.}
Human verification inherently caps the rate of engagement accumulation. PoCmt is
therefore not suitable for extremely high-throughput settings without
hybridization.

\subsection{Future Directions}

PoCmt opens multiple research avenues:

\begin{itemize}
	\item equilibrium modeling of incentives under fatigue, churn, and labor markets,
	\item adaptive and diverse AI-resistant challenge generation,
	\item multi-resource consensus combining human-time with stake, storage, or other signals,
	\item formal models of collusion and coordinated human groups,
	\item prototype deployments to quantify participation, usability, and noise.
\end{itemize}

Overall, PoCmt expands the design space of permissionless consensus by leveraging
an intrinsically scarce and non-parallelizable resource---human attention---yet
achieving its full potential requires addressing the practical, economic, and
behavioral complexities outlined above.

\section{Conclusion}

This work introduced \emph{Proof of Commitment} (PoCmt), a new
permissionless consensus primitive grounded in real-time,
human-verifiable engagement rather than machine-parallelizable
resources. We formalized a decomposable commitment state
$S_v(t) = (H_v(t), P_v(t), U_v(t))$ whose evolution is governed by
windowed engagement boosts, availability decay, and slashing for
protocol violations. This structure ensures that validator influence
accumulates only through sustained human effort and degrades
predictably under inactivity or misbehavior.

Through the Human Challenge Oracle (HCO), we modeled human-time as an
explicit, non-parallelizable resource and established a sharp
cost-theoretic separation from classical paradigms. We proved that any
adversary maintaining $s$ active identities must invest
$\Theta(s)$ units of human-time per window, and that adversarial
commitment cannot outgrow honest commitment without continuously paying
this linear cost. In contrast, we showed that consensus protocols based
on divisible and parallelizable resources—such as Proof of Work and
Proof of Stake—cannot enforce linear Sybil cost under their standard
assumptions.

Building on this resource model, we presented a complete PoCmt
consensus protocol, including commitment-weighted leader election,
block proposal and validation, fork choice, and optional BFT-style
finality. Our security analysis established backbone-style guarantees
of safety, liveness, and proportional fairness under partial synchrony,
as well as inherent robustness against long-range and ``zombie''
attacks via time-dependent commitment decay. A simulation study
complemented the analysis, empirically validating the commitment drift
invariant, the linear dependence of adversarial influence on
human-time capacity, and the fairness of commitment-proportional leader
selection.

PoCmt expands the design space of permissionless consensus by grounding
security in an intrinsically scarce, non-amortizable resource—human
attention—rather than computation or capital. While this shift raises
new practical questions around incentives, usability, and challenge
design, it demonstrates that strong Sybil resistance need not rely on
machine-dominated resources. By reframing consensus around human-time,
PoCmt opens a new pathway toward decentralized protocols whose security
reflects genuine, sustained human participation rather than technical
or financial amplification.

\bibliographystyle{ACM-Reference-Format}
\bibliography{refs}

\end{document}